\documentclass{llncs}
\usepackage[english]{babel}
\usepackage{amsmath,amsfonts}
\usepackage{enumerate}
\usepackage{graphicx}
\usepackage{color}
\usepackage{psfrag}

\usepackage{amssymb,amsbsy,latexsym}

\newcommand{\COMMENTED}[1]{}

\newcommand{\R}{\mathcal{R}}
\newcommand{\C}{\mathcal{C}}
\newcommand{\F}{\mathcal{F}}
\newcommand{\SSS}{\mathcal{S}}

\newcommand{\E}{{\mbox {\bf E}}}

\newcommand{\Rc}{\ensuremath{\mathcal R}}
\newcommand{\rsum}{\ensuremath{\text{Sum}}}
\newcommand{\bx}{\ensuremath{\overline{x}}}
\newcommand{\br}{\ensuremath{\overline{r}}}
\newcommand{\sm}{\ensuremath{\setminus}}

\begin{document}

\title{Fault-Tolerant Facility Location:\\ a randomized dependent LP-rounding
  algorithm\thanks{This work was partially supported by: 
(i) the Future and Emerging Technologies Unit of EC (IST priority - 6th FP), under contract no.\ FP6-021235-2 (project ARRIVAL), 
(ii) MNISW grant number N N206 1723 33, 2007–-2010,
(iii)  NSF ITR Award CNS-0426683 and NSF Award CNS-0626636, and
(iv) NSERC grant 327620-09 and an Ontario Early Researcher Award.}}

\author{ Jaroslaw Byrka\inst{1}\thanks{
Work of this author was partially conducted at CWI Amsterdam, TU Eindhoven, and while visiting
the University of Maryland.} 
\and Aravind Srinivasan\inst{2}
\and Chaitanya Swamy\inst{3} 
}

\institute{
Institute of Mathematics, Ecole Polytechnique Federale de Lausanne,\\
CH-1015 Lausanne, SWITZERLAND. \email{jaroslaw.byrka@epfl.ch}
\and
Dept.\ of Computer Science and
Institute for Advanced Computer Studies, University of Maryland,
College Park, MD 20742, USA. \email{srin@cs.umd.edu}
\and
Dept.\ of Combinatorics \& Optimization,
Faculty of Mathematics,
University of Waterloo,
Waterloo, ON N2L 3G1, CANADA. \email{cswamy@math.uwaterloo.ca}
}

\maketitle
\begin{abstract}
We give a new randomized LP-rounding $1.725$-approximation algorithm for the metric Fault-Tolerant
Uncapacitated Facility Location problem. This improves on the previously best known
$2.076$-approximation algorithm of Swamy \& Shmoys. To the best of our knowledge, our work 
provides the first application of
a dependent-rounding technique in the domain
of facility location. The analysis of our algorithm benefits
from, and extends,
methods developed for Uncapacitated Facility Location; it also helps
uncover new properties of the dependent-rounding approach.

An important concept that we develop is a novel, hierarchical clustering scheme.
Typically, LP-rounding approximation algorithms for facility location problems are based on
partitioning facilities into disjoint clusters and opening at least one facility in each cluster.
We extend this approach and construct a laminar family of clusters, 
which then guides the rounding procedure. It allows to exploit properties of 
dependent rounding, and provides a quite tight analysis resulting in the 
improved approximation ratio.
\end{abstract}



\section{Introduction}
In Facility Location problems we are given a set of clients $\C$
that require a certain service. To provide such a service, 
we need to open a subset of a given set of facilities $\F$.
Opening each facility $i \in \F$ costs $f_i$ and serving
a client $j$ by facility $i$ costs $c_{ij}$; the standard
assumption is that the $c_{ij}$ are symmetric and constitute a
metric. (The non-metric case is much harder to approximate.)
In this paper, we follow Swamy \& Shmoys \cite{DBLP:journals/talg/SwamyS08}
and study
the Fault-Tolerant Facility Location (FTFL) problem, where each client
has a positive integer specified as its
\textit{coverage requirement} $r_j$. The task is to find a minimum-cost
solution which opens some facilities from $\F$ and connects each client
$j$ to $r_j$ different open facilities.

The FTFL problem was introduced by Jain \& Vazirani~\cite{DBLP:journals/algorithmica/JainV03}. Guha et al.~\cite{DBLP:journals/jal/GuhaMM03} gave the first constant factor approximation algorithm
with approximation ratio 2.408. This was later improved by Swamy \& Shmoys~\cite{DBLP:journals/talg/SwamyS08}
who gave a 2.076-approximation algorithm.
FTFL generalizes the standard Uncapacitated Facility Location (UFL)
problem wherein $r_j = 1$ for all $j$, for which Guha \& Khuller~\cite{DBLP:journals/jal/GuhaK99} proved an approximation lower bound of 
$\approx$ 1.463. The current-best approximation ratio for UFL is achieved 
by the 1.5-approximation algorithm of Byrka~\cite{DBLP:conf/approx/Byrka07}.

In this paper we give a new LP-rounding 
$1.7245$-approximation algorithm for the FTFL problem.
It is the first application of the dependent rounding technique of \cite{DBLP:conf/focs/Srinivasan01}
to a facility location problem. 

Our algorithm uses a novel clustering method, 
which allows clusters not to be disjoint, but rather
to form a laminar family of subsets of facilities.
The hierarchical structure of the obtained clustering
exploits properties of dependent rounding. 
By first rounding the ``facility-opening'' variables within smaller clusters,
we are able to ensure that at least a certain number of facilities is open
in each of the clusters. Intuitively, by allowing clusters to have different
sizes we may, in a more efficient manner, guarantee the opening of
sufficiently-many facilities
around clients with different coverage requirements $r_j$. 
In addition, one of our main technical contributions 
is Theorem~\ref{dep-vs-indep_thm}, which develops
a new property of the dependent-rounding technique that appears likely
to have further applications. Basically, suppose we apply
dependent rounding to a sequence of reals and consider an arbitrary
subset $S$ of the rounded variables (each of which lies in $\{0,1\}$) as
well as an arbitrary integer $k > 0$. 
Then, a natural fault-tolerance-related objective is that if $X$ denotes the 
number of variables rounded to $1$ in $S$, then the random variable
$Z = \min\{k, X\}$ be ``large''. (In other words, we want $X$ to be ``large'',
but $X$ being more than $k$ does not add any marginal utility.)
We prove that if $X_0$ denotes the corresponding
sum wherein the reals are rounded \emph{independently} and if
$Z_0 = \min\{k, X_0\}$, then $\E[Z] \geq \E[Z_0]$. Thus, for analysis purposes,
we may work with $Z_0$, which is much more tractable due to the independence;
at the same time, we derive all the benefits of dependent rounding (such as
a given number of facilities becoming available in a cluster, with probability
one). Given the growing number of applications of dependent-rounding
methodologies, we view this as a useful addition to the toolkit. 

\section{Dependent rounding}
\label{sec:dep-round}
Given a fractional vector $y = (y_1, y_2, \ldots, y_N) \in [0,1]^N$
we often seek to round it
to an integral vector $\hat{y} \in \{0,1\}^N$ that is in a
problem-specific sense very ``close to'' $y$. The
dependent-randomized-rounding technique of \cite{DBLP:conf/focs/Srinivasan01}
is one such approach known for preserving the sum of the entries
deterministically, along with concentration bounds for any linear
combination of the entries; we will generalize a known property of this
technique in order to apply it to the FTFL problem. The very useful
\emph{pipage rounding} technique of \cite{ageev-sviri} was developed
prior to \cite{DBLP:conf/focs/Srinivasan01}, and can be viewed as a
derandomization (deterministic analog) of \cite{DBLP:conf/focs/Srinivasan01}
via the method of conditional probabilities. Indeed, the results of
\cite{ageev-sviri} were applied in the work of
\cite{DBLP:journals/talg/SwamyS08}; the probabilistic intuition, as well as
our generalization of the analysis of \cite{DBLP:conf/focs/Srinivasan01},
help obtain our results.

Define $[t] = \{1, 2, \ldots, t\}$.
Given a fractional vector $y = (y_1, y_2, \ldots, y_N) \in [0,1]^N$,
the rounding technique of \cite{DBLP:conf/focs/Srinivasan01} (henceforth
just referred to as ``dependent rounding'') is a polynomial-time randomized
algorithm to produce a random vector $\hat{y} \in \{0,1\}^N$ with the
following three properties:
\begin{description}
\item[(P1): marginals.] $\forall i, ~\Pr[\hat{y}_i = 1] = y_i$;
\item[(P2): sum-preservation.]
With probability one, $\sum_{i=1}^N \hat{y}_i$ equals
either $\lfloor \sum_{i=1}^N y_i \rfloor$ or $\lceil \sum_{i=1}^N y_i \rceil$;
and
\item[(P3): negative correlation.] $\forall S \subseteq [N]$,
$\Pr[\bigwedge_{i \in S}(\hat{y}_i = 0)] \leq \prod_{i \in S}(1 - y_i)$, and
$\Pr[\bigwedge_{i \in S}(\hat{y}_i = 1)] \leq \prod_{i \in S} y_i$.
\end{description}
The dependent-rounding algorithm is described in Appendix~\ref{app:dep-round}.
In this paper, we also exploit the order in which the entries of the given 
fractional vector $y$ are rounded. We initially define a laminar family of subsets of indices
$\SSS \subseteq 2^{[N]}$. When applying the dependent rounding procedure, 
we first round within the smaller sets,
until at most one fractional entry in a set is left, then we proceed with bigger sets
possibly containing the already rounded entries. It can easily be shown
that it assures the following version of property (P2) for all subsets $S$
from the laminar family $\SSS$:
\begin{description}
\item[(P2'): sum-preservation.] With probability one, 
$\sum_{i \in S} \hat{y}_i = \sum_{i \in S} y_i$\\ and 
$|\{i \in S: \hat{y}_i = 1 \}| = \lfloor \sum_{i \in S} y_i \rfloor$.
\end{description}

Now, let $S \subseteq [N]$ be any subset, not necessarily 
from $\SSS$. In order to present our results, we
need two functions, $\mbox{Sum}_S$ and $g_{\lambda, S}$.
For any vector $x \in [0,1]^n$, let $\mbox{Sum}_S(x) =
\sum_{i \in S} x_i$ be the sum of the
elements of $x$ indexed by elements of $S$; in particular, if
$x$ is a (possibly random) vector with all entries either $0$ or $1$,
then $\mbox{Sum}_S(x)$ counts the number of entries in $S$ that are $1$.
Next, given $s = |S|$ and a real vector
$\lambda = (\lambda_0, \lambda_1, \lambda_2, \ldots, \lambda_s)$, we
define, for any $x \in \{0,1\}^n$,
\[ g_{\lambda,S}(x) =
\sum_{i=0}^s \lambda_i \cdot \mathcal{I}(\mbox{Sum}_S(x) = i), \]
where $\mathcal{I}(\cdot)$ denotes the indicator function.
Thus, $g_{\lambda,S}(x) = \lambda_i$ if $\mbox{Sum}_S(x) = i$.

Let $\mathcal{R}(y)$ be a random vector in $\{0,1\}^N$ obtained by
\emph{independently} rounding
each $y_i$ to $1$ with probability $y_i$, and to $0$ with the
complementary probability of $1 - y_i$.
Suppose, as above, that $\hat{y}$ is a random vector in $\{0,1\}^N$ obtained by
applying the dependent rounding technique to $y$.
We start with a general theorem
and then specialize it to Theorem~\ref{dep-vs-indep_thm} that
will be very useful for us:

\begin{theorem}
\label{dep-vs-indep:general}
Suppose we conduct dependent rounding on $y = (y_1, y_2, \ldots, y_N)$.
Let $S \subseteq [N]$ be any subset with cardinality $s \geq 2$, and
let $\lambda = (\lambda_0, \lambda_1, \lambda_2, \ldots, \lambda_s)$ be
any vector, such that for all $r$ with $0 \leq r \leq s-2$ we have
$\lambda_{r} - 2 \lambda_{r + 1} + \lambda_{r + 2} \leq 0$. Then,
$\E[g_{\lambda,S}(\hat{y})] \geq \E[g_{\lambda,S}(\mathcal{R}(y))]$.
\end{theorem}

\begin{theorem}
\label{dep-vs-indep_thm}
  For any $y \in [0,1]^N$, $S \subseteq [N]$, and $k = 1,2, \ldots$, we have
\[
\E[\min\{k, \mbox{Sum}_S(\hat{y})\}] \geq
\E[\min\{k, \mbox{Sum}_S(\mathcal{R}(y))\}].
\]
\end{theorem}

Using the notation $\exp(t) = e^t$, our next key result is:

\begin{theorem}
\label{dr_thm}
  For any $y \in [0,1]^N$, $S \subseteq [N]$, and $k = 1,2, \ldots$, we have
\[
\E[\min\{k, \mbox{Sum}_S(\mathcal{R}(y))\}] \geq k \cdot (1 - \exp(-\mbox{Sum}_S(y) / k)).
\]
\end{theorem}

The above two theorems yield a key corollary that we will use:
\begin{corollary}
\label{cor:dr_cor}
\[
\E[\min\{k, \mbox{Sum}_S(\hat{y})\}] \geq
k \cdot (1 - \exp(-\mbox{Sum}_S(y) / k)).
\]
\end{corollary}

%


Proofs of the theorems from this section are provided in Appendix B.

\section{Algorithm}
\subsection{LP-relaxation}
The FTFL problem is defined by the following
Integer Program (IP).
\begin{eqnarray}
 \label{eq1}\mbox{minimize} &\sum_{i \in \F} f_i y_i + \sum_{j \in \C} \sum_{i \in \F} c_{ij}x_{ij} \\
 \mbox{subject to:}&\sum_i  x_{ij} \geq r_j  & \forall j \in \C\\
	& x_{ij} \leq y_i & \forall j \in \C \; \forall i \in \F\\
	& y_i \leq 1  & \forall i \in \F\\
\label{integrality}	& x_{ij},y_i \in Z_{\geq 0} & \forall j \in \C \; \forall i \in \F,
\end{eqnarray}
where $\C$ is the set of clients, $\F$ is the set of possible locations of facilities,
$f_i$ is a cost of opening a facility at location $i$, $c_{ij}$ is a cost of serving
client $j$ from a facility at location $i$, and $r_j$ is the amount of facilities
client $j$ needs to be connected to.

If we relax constraint~(\ref{integrality}) to $x_{ij},y_i \geq 0$ we obtain the standard LP-relaxation
of the problem. Let $(x^*,y^*)$ be an optimal solution to this LP relaxation.
We will give an algorithm that rounds this solution to an integral solution $(\tilde{x},\tilde{y})$ with cost
at most $\gamma \approx 1.7245$ times the cost of $(x^*,y^*)$.

\subsection{Scaling} \label{sec_scaling}

We may assume, without loss of generality, that for any client $j \in \C$
there exists at most one facility $i \in \F$ such that $0 < x_{ij} < y_i$. Moreover, this
facility may be assumed to have the highest distance to client $j$ among the
facilities that fractionally serve $j$ in $(x^*,y^*)$.

We first set $\tilde{x}_{ij} = \tilde{y}_i = 0$ for all $i \in \F$, $j \in \C$.
Then we scale up the fractional solution by the constant $\gamma \approx 1.7245$
to obtain a fractional solution $(\hat{x},\hat{y})$. To be precise:
we set $\hat{x}_{ij} = \min\{1,\gamma \cdot x^*_{ij}\}$,
$\hat{y}_i = \min\{1,\gamma \cdot y^*_i\}$.
We open each facility $i$ with $\hat{y}_i = 1$ and connect each client-facility pair with $\hat{x}_{ij} = 1$.
To be more precise, we modify $\hat{y}$, $\tilde{y}$, $\hat{x}$, $\tilde{x}$ and service requirements $r$ as follows.
For each facility $i$ with $\hat{y}_i = 1$, set $\hat{y}_i = 0$ and $\tilde{y}_i = 1$.
Then, for every pair $(i,j)$ such that $\hat{x}_{ij} = 1$,
set $\hat{x}_{ij} = 0$, $\tilde{x}_{ij} = 1$  and decrease $r_j$ by one. When this process is finished
we call the resulting $r$, $\hat{y}$ and $\hat{x}$ by $\overline{r}$, $\overline{y}$ and $\overline{x}$.
Note that the connections that we made in this phase may be paid for
by a difference in the connection cost between $\hat{x}$ and $\overline{x}$.
We will show that the remaining connection cost of the solution of the algorithm
is expected to be at most the cost of $\overline{x}$.

For the feasibility of the final solution, it is essential that if we connected client $j$
to facility $i$ in this initial phase,
we will not connect it again to $i$ in the rest of the algorithm.
There will be two ways of connecting clients in the process of
rounding $\overline{x}$. The first one connects client $j$
to a subset of facilities serving $j$ in $\overline{x}$.
Recall that if $j$ was connected to facility $i$ in the initial phase,
then $\overline{x}_{ij} = 0$, and no additional $i$-$j$ connection will be created.

The connections of the second type will be created in a process of \emph{clustering}.
The clustering that we will use is a generalization of the clustering used
by Chudak \& Shmoys for the UFL
problem \cite{DBLP:journals/siamcomp/ChudakS03}.
As a result of this clustering process,
client $j$ will be allowed to connect itself via a different client $j'$
to a facility open around $j'$. $j'$ will be called a \emph{cluster center} for
a subset of facilities, and it will make sure that at least some
guaranteed number of these facilities will get opened.

To be certain that client $j$ does not get again connected to facility $i$
with a path via client $j'$, facility $i$ will never
be a member of the set of facilities clustered by client $j'$.
We call a facility $i$ \emph{special} for client $j$ iff
$\tilde{y}_i=1$ and $0 < \overline{x}_{ij} < 1$. Note that, by our earlier
assumption, there is at most one special facility for each client $j$,
and that a special facility must be at maximal distance among facilities
serving $j$ in $\overline{x}$. When rounding the fractional solution in
Section~\ref{sec_rounding}, we take care that special facilities are not
members of the formed clusters.

\subsection{Close and distant facilities}

Before we describe how do we cluster facilities,
we specify the facilities that are interesting for a particular
client in the clustering process. The following can be fought of
as a version of a \emph{filtering} technique of Lin and Vitter~\cite{DBLP:conf/stoc/LinV92},
first applied to facility location by Shmoys et al.~\cite{DBLP:conf/stoc/ShmoysTA97}.
The analysis that we use here is a version of the argument of Byrka~\cite{DBLP:conf/approx/Byrka07}.

As a result of the scaling that was described in the previous section,
the connection variables $\overline{x}$ amount for a total connectivity
that exceeds the requirement $\overline{r}$.
More precisely, we have $\sum_{i \in \F}\overline{x}_{ij} \geq \gamma \cdot \overline{r}_j$
for every client $j \in \C$. We will consider for each client $j$ a subset of facilities
that are just enough to provide it a fractional connection of $\overline{r}_j$.
Such a subset is called a set of \emph{close facilities} of client $j$ and is defined as follows.

For every client $j$ consider the following construction.
Let $i_1, i_2, \ldots, i_{|\F|}$ be the ordering of facilities in $\F$ in
a nondecreasing order of distances $c_{ij}$ to client $j$.
Let $i_k$ be the facility in this ordering, such that $\sum_{l=1}^{k-1} \overline{x}_{i_{l}j} < \overline{r}_j$
and $\sum_{l=1}^{k} \overline{x}_{i_{l}j} \geq \overline{r}_j$.
Define
\[
\overline{x}_{i_{l}j}^{(c)} = \left\{ \begin{array}{ll}
 \overline{x}_{i_{l}j} & \mbox{ for } l<k,\\
 \overline{r}_j - \sum_{l=1}^{k-1} \overline{x}_{i_{l}j} & \mbox{ for } l=k,\\
 0 & \mbox{ for } l>k
\end{array}\right.
\]
Define $\overline{x}_{ij}^{(d)} = \overline{x}_{ij} - \overline{x}_{ij}^{(c)}$ for all $i \in \F, j \in \C$.

We will call the set of facilities $i \in \F$ such that $\overline{x}_{ij}^{(c)} > 0$ the set of \emph{close
facilities} of client $j$ and we denote it by $C_j$. By analogy, we will call the set of facilities $i \in \F$
such that $\overline{x}_{ij}^{(d)} > 0$ the set of \emph{distant facilities} of client $j$ and denote it $D_j$.
Observe that for a client $j$ the intersection of $C_j$ and $D_j$ is either empty, or contains exactly one facility.
In the latter case, we will say that this facility is both distant and close. Note that, unlike in the UFL problem,
we may not simply split this facility to the close and the distant part, because
it is essential that we make at most one
connection to this facility in the final integral solution.
Let $d^{(max)}_j = c_{i_k j}$ be the distance from client $j$ to the farthest of its close facilities.

\subsection{Clustering}
\label{sec_clustering}
We will now construct a family of subsets of facilities $\SSS \in 2^\F$.
These subsets $S \in \SSS$ will be called clusters and
they will guide the rounding procedure described next.
There will be a client related to each cluster, and each single
client $j$ will be related to at most one cluster, which we call $S_j$.

Not all the clients participate in the clustering process. 
Clients $j$ with $\overline{r}_j = 1$ and
a special facility $i' \in C_j$ (recall that a special
facility is a facility that is fully open in $\hat{y}$ but only partially used
by $j$ in $\overline{x}$) will be called special and will not take part in the clustering process.
Let $\C'$ denote the set of all other, non-special clients.
Observe that, as a result of scaling, clients $j$ with $\overline{r}_j \geq 2$ do not have any special facilities among
their close facilities (since $\sum_i\bx_{ij} \geq \gamma \br_j > \br_j + 1$). As a consequence, there are no special facilities
among the close facilities of clients from $\C'$, the only clients actively
involved in the clustering procedure.

For each client $j \in \C'$ we will keep two families $A_j$ and $B_j$ of
disjoint subsets of facilities.
Initially $A_j = \{ \{i\} : i \in C_j \}$, i.e., $A_j$
is initialized to contain a singleton set for each close facility
of client $j$; $B_j$ is initially empty.
$A_j$ will be used to store these initial singleton sets, 
but also clusters containing only close facilities of $j$;
$B_j$ will be used to store only clusters that contain at least one
close facility of $j$. When adding a cluster to either $A_j$ or $B_j$
we will remove all the subsets it intersects from both $A_j$ and $B_j$,
therefore subsets in $A_j \cup B_j$ will always be pairwise disjoint.

The family of clusters that we will construct will be a laminar family of subsets
of facilities, i.e., any two clusters are either disjoint or one entirely contains the other.
One may imagine facilities being leaves and clusters being internal nodes
of a forest that eventually becomes a tree, when all the clusters are added.

We will use $\overline{y}(S)$ as a shorthand for $\sum_{i \in S}\overline{y}_i$.
Let us define $\underline{y}(S) = \lfloor \overline{y}(S) \rfloor $.  
As a consequence of using the family of clusters to guide the rounding process, 
by Property (P2') of the dependent rounding procedure when applied to a cluster, 
th quantity $\underline{y}(S)$ lower bounds the number of facilities that will certainly be opened in cluster $S$.
Additionally, let us define the residual requirement of client $j$ to be $rr_j = \overline{r}_j - \sum_{S \in (A_j \cup B_j)} \underline{y}(S)$, that is $\overline{r}_j$ minus a lower bound on the number of facilities that will be opened
in clusters from $A_j$ and $B_j$.

We use the following procedure to compute clusters.
While there exists a client $j \in \C'$, such that $rr_j > 0$,
take such $j$ with minimal $d_j^{(max)}$ and do the following:
\begin{enumerate}
 \item Take $X_j$ to be an inclusion-wise minimal subset of $A_j$,
       such that $\sum_{S \in X_j} (\overline{y}(S) - \underline{y}(S)) \geq rr_j$.
       Form the new cluster $S_j= \bigcup_{S \in X_j} S$.
 \item Make $S_j$ a new cluster by setting $\SSS \leftarrow \SSS \cup \{S_j\}$.
 \item Update $A_{j} \leftarrow (A_{j} \setminus X_j) \cup \{S_j\}$.
 \item For each client $j'$ with $rr_{j'} > 0$ do
\begin{itemize}
 \item If $X_j \subseteq A_{j'}$,
 	then set $A_{j'} \leftarrow (A_{j'} \setminus X_j) \cup \{S_j\}$.
 \item If $X_j \cap A_{j'} \neq \emptyset$ and $X_j \setminus A_{j'} \neq \emptyset$,\\
	then set $A_{j'} \leftarrow A_{j'} \setminus X_j$ and $B_{j'} \leftarrow \{S \in B_{j'}: S \cap S_j = \emptyset\} \cup \{S_j\}$.
\end{itemize}
\end{enumerate}
Eventually, add a cluster $S_r = \F$ containing all the facilities to the family $\SSS$.

We call a client $j'$ active in a particular iteration, if before this iteration
its residual requirement $rr_j = \overline{r}_j - \sum_{S \in (A_j \cup B_j)} \underline{y}(S) $ was positive.
During the above procedure, all active clients $j$ have in their sets $A_j$ and $B_j$ only maximal subsets of facilities, that means they are not subsets of any other clusters (i.e., they are roots of their trees in the current forest). Therefore, when a new cluster $S_j$ is created, it contains all the other clusters with which it has nonempty intersections (i.e., the new cluster $S_j$ becomes a root of a new tree).

We shall now argue that there is enough fractional opening in clusters in $A_j$
to cover the residual requirement $rr_j$ when cluster $S_j$ is to be formed. Consider a fixed client $j \in \C'$.
Recall that at the start of the clustering we have $A_j = \{ \{i\}: i \in C_j\}$, and therefore 
$\sum_{S \in A_j}( \overline{y}(S)- \underline{y}(S)) = \sum_{i \in C_j} \overline{y}_{i} \geq \overline{r}_j = rr_j$.
It remains to show, that $\sum_{S \in A_j} (\overline{y}(S) - \underline{y}(S)) - rr_j$ does not decrease over time
until client $j$ is considered.
When a client $j'$ with $d_{j'}^{(max)} \leq d_j^{(max)}$ is considered and cluster $S_{j'}$ is created, the following cases are possible:
\begin{enumerate}
\item $S_{j'} \cap (\bigcup_{S \in A_j}S) = \emptyset$: then $A_j$ and $rr_j$ do not change;
\item $S_{j'} \subseteq (\bigcup_{S \in A_j}S)$: then $A_j$ changes its structure, but $\sum_{S \in A_j} \overline{y}(S)$ and $\sum_{S \in B_j} \underline{y}(S)$ do not change; hence $\sum_{S \in A_j}( \overline{y}(S) - \underline{y}(S)) - rr_j$ also does not change;
\item $S_{j'} \cap (\bigcup_{S \in A_j}S) \neq \emptyset$ and $S_{j'} \setminus (\bigcup_{S \in A_j}S) \neq \emptyset$: then, by inclusion-wise minimality of set $X_{j'}$, we have \mbox{$ \underline{y}(S_{j'}) - \sum_{S \in B_j, S \subseteq S_{j'}} \underline{y}(S) - \sum_{S \in A_j, S \subseteq S_{j'} } \overline{y}(S) \geq 0$;} hence, $\sum_{S \in A_j} (\overline{y}(S) - \underline{y}(S)) - rr_j$ cannot decrease.
\end{enumerate}


Let $A'_j = A_j \cup \SSS$ be the set of clusters in $A_j$.
Recall that all facilities in clusters in $A'_j$ are close facilities of $j$. 
Note also that each cluster $S_{j'} \in B_j$ was created from close facilities of a client $j'$
with $d_{j'}^{(max)} \leq d_{j}^{(max)}$. We also have for each $S_{j'} \in B_j$ that $S_{j'} \cap C_j \neq \emptyset$,
hence, by the triangle inequality, all facilities in $S_{j'}$ are at distance at most $3 \cdot d_{j}^{(max)}$
from $j$. We thus infer the following

\begin{corollary} \label{cor_clustering}
The family of clusters $\SSS$ contains for each client $j \in \C'$ a collection of disjoint clusters
$A'_j \cup B_j$ containing only facilities within distance $3 \cdot d_{j}^{(max)}$, 
and $\sum_{S \in A'_j \cup B_j} \lfloor \sum_{i \in S} \overline{y}_i \rfloor \geq \overline{r}_j$.
\end{corollary}

Note that our clustering is related to, 
but more complex then the one of Chudak and Shmoys~\cite{DBLP:journals/siamcomp/ChudakS03} for UFL
and of Swamy and Shmoys~\cite{DBLP:journals/talg/SwamyS08} for FTFL,
where clusters are pairwise disjoint and each contains facilities
whose fractional opening sums up to or slightly exceeds the value of 1.

\subsection{Opening of facilities by dependent rounding}
\label{sec_rounding}

Given the family of subsets $\SSS \in 2^\F$ computed by the clustering procedure
from Section~\ref{sec_clustering}, we may proceed with rounding the 
fractional opening vector $\overline{y}$ into an integral vector $y^R$.
We do it by applying the rounding technique of Section~\ref{sec:dep-round},
guided by the family $\SSS$, which is done as follows.

While there is more than one fractional entry, select a minimal subset of 
$S \in \SSS$ which contains more than one fractional entry and apply the rounding procedure to entries of $\overline{y}$ indexed by elements of $S$ until at most one entry in $S$ remains fractional.
Eventually, if there remains a fractional entry, round it independently and let $y^R$ be the resulting vector.

Observe that the above process is one of the possible implementations of dependent rounding
applied to $\overline{y}$. As a result, the random integral vector $y^R$ satisfies properties (P1),(P2), and (P3).
Additionally, property (P2') holds for each cluster $S \in \SSS$.
Hence, at least $\lfloor \sum_{i \in S} \overline{y}_i \rfloor$ entries in each $S \in \SSS$ are rounded to 1.
Therefore, by Corollary~\ref{cor_clustering}, we get
\begin{corollary} \label{cor_rounding}
For each client $j \in \C'$.
\[
|\{i \in \F | y^R_i = 1 \mbox{ and } c_{ij} \leq 3 \cdot d_{j}^{(max)} \}| \geq \overline{r}_j.
\]
\end{corollary}

Next, we combine the facilities opened by rounding $y^R$ with facilities opened already
when scaling which are recorded in $\tilde{y}$, i.e., we update $\tilde{y} \leftarrow \tilde{y} + y^R$.

Eventually, we connect each client $j \in \C$ to $r_j$ closest opened facilities and code it in $\tilde{x}$.

\section{Analysis}
We will now estimate the expected cost of the solution $(\tilde{x},\tilde{y})$.
The tricky part is to bound the connection cost, which we do as follows.
We argue that a certain fraction of the demand of client $j$ may
be satisfied from its close facilities, then some part of the remaining
demand can be satisfied from its distant facilities.
Eventually, the remaining (not too large in expectation) part of the demand
is satisfied via clusters.


\subsection{Average distances}
Let us consider weighted average distances from a client $j$ to sets of facilities fractionally serving it.
Let $d_j$ be the average connection cost in $\overline{x}_{ij}$ defined as
\[
d_j = \frac{\sum_{i \in \F} c_{ij} \cdot \overline{x}_{ij} }{\sum_{i \in \F} \overline{x}_{ij}}.
\]
Let $d^{(c)}_j$, $d^{(d)}_j$  be the average connection costs in $\overline{x}_{ij}^{(c)}$ and $\overline{x}_{ij}^{(d)}$ defined as
\[
d^{(c)}_j = \frac{\sum_{i \in \F} c_{ij} \cdot \overline{x}_{ij}^{(c)} }{\sum_{i \in \F} \overline{x}_{ij}^{(c)}} ,
\]
\[
d^{(d)}_j = \frac{\sum_{i \in \F} c_{ij} \cdot \overline{x}_{ij}^{(d)} }{\sum_{i \in \F} \overline{x}_{ij}^{(d)}} .
\]
Let $R_j$ be a parameter defined as
\[
 R_j = \frac{d_j - d^{(c)}_j}{d_j}
\]
if $d_j > 0$ and $R_j = 0$ otherwise.
Observe that $R_j$ takes value between $0$ and $1$.
$R_j = 0$ implies $d^{(c)}_j = d_j = d^{(d)}_j$, and
$R_j = 1$ occurs only when $d^{(c)}_j = 0$. The role played by $R_j$
is that it measures a certain parameter of the instance, big values
are good for one part of the analysis, small values are good for the other.

\begin{lemma} \label{average_lemma}
  $d^{(d)}_j \leq d_j (1+ \frac{R_j}{\gamma -1})$.
\end{lemma}
\begin{proof}
Recall that $\sum_{i \in \F} \overline{x}_{ij}^{(c)} = \overline{r}_j$ and
$\sum_{i \in \F} \overline{x}_{ij}^{(d)} \geq (\gamma-1) \cdot \overline{r}_j$.
Therefore, we have $(d^{(d)}_j - d_j)\cdot(\gamma -1) \leq (d_j - d^{(c)}_j) \cdot 1 = R_j \cdot d_j$,
which can be rearranged to get $d^{(d)}_j \leq d_j (1+ \frac{R_j}{\gamma -1})$.
\end{proof}

\begin{figure}[t]
\begin{center}
\fboxsep7pt
\framebox[0.90\columnwidth]{
\begin{minipage}{0.80\columnwidth}
\begin{center}

\psfrag{g}{$\frac{\overline{r}_j}{\gamma}$}
\psfrag{r}{$\overline{r}_j$}
\psfrag{c1}{$d^{(d)}_j$} 
\psfrag{c2}{$\hspace{-2mm} d^{(max)}_j$} 
\psfrag{c3}{$d_j$} 
\psfrag{c4}{$\hspace{-5mm} d^{(c)}_j = d_j(1 - R_j)$} 

\includegraphics[width=9cm]{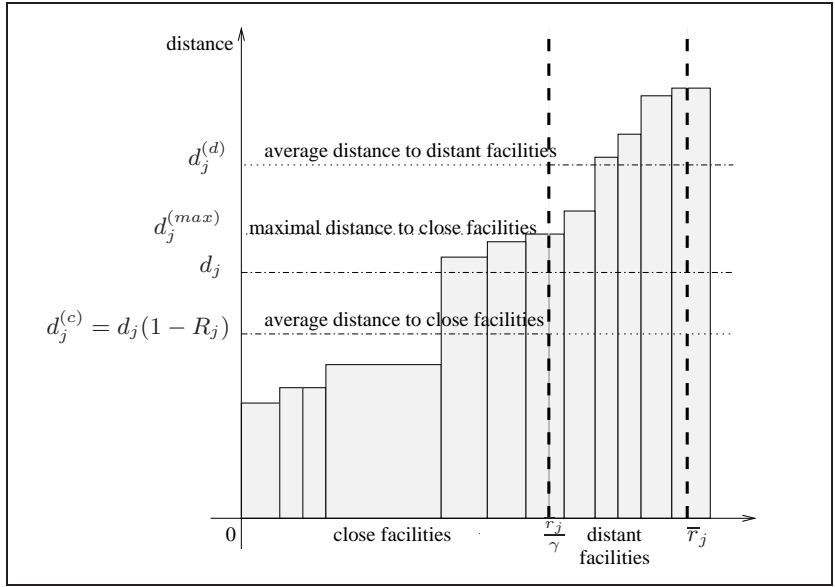}
\end{center}
\end{minipage}
}
\end{center}
\caption{Distances to facilities serving client $j$ in $\overline{x}$. The width of a rectangle corresponding to facility $i$
is equal to $\overline{x}_{ij}$. Figure helps to understand the meaning of $R_j$.}
\label{fig:distances}
\end{figure}

Finally, observe that the average distance from $j$ to the distant facilities of $j$ gives an upper bound
on the maximal distance to any of the close facilities of $j$. Namely, $d^{(max)}_j \leq d^{(d)}_j$.
\subsection{Amount of service from close and distant facilities}
We now argue that in the solution $(\tilde{x},\tilde{y})$,
a certain portion of the demand is expected to be served by the close
and distant facilities of each client.
Recall that for a client $j$ it is possible, that there is
a facility that is both its close and its distant facility.
Once we have a solution that opens such a facility,
we would like to say what fraction of the demand is served from the close
facilities. To make our analysis simpler we will toss a properly biased coin
to decide if using this facility counts as using a close facility.
With this trick we, in a sense, split such a facility into a close
and a distant part. Note that we may only do it for this part of
the analysis, but not for the actual rounding algorithm from Section~\ref{sec_rounding}.
Applying the above-described split of the undecided facility,
we get that the total fractional opening of close facilities of client $j$
is exactly $\overline{r}_j$, and the total fractional opening of both close
and distant facilities is at least $\gamma \cdot \overline{r}_j$.
Therefore, Corollary~\ref{cor:dr_cor} yields the following:

\begin{corollary} \label{cor_1}
  The amount of close facilities used by client $j$ in a solution
  described in Section~\ref{sec_rounding} is expected to be at least $(1-\frac{1}{e}) \cdot \overline{r}_j$.
\end{corollary}

\begin{corollary}\label{cor_2}
  The amount of close and distant facilities used by client $j$ in a solution
  described in Section~\ref{sec_rounding} is expected to be at least $(1-\frac{1}{e^\gamma}) \cdot \overline{r}_j$.
\end{corollary}

Motivated by the above bounds we design a selection method to choose a 
(large-enough in expectation)
subset of facilities opened around client $j$:

\begin{lemma} \label{lem:facility_selection}
For $j \in \C'$ we can select a subset $F_j$ of open facilities 
from $C_j \cup D_j$ such that:
 \begin{eqnarray*}
  |F_j| & \leq & \overline{r}_j \mbox{ (with probability 1)},\\
 E[F_j] & = & (1-\frac{1}{e^\gamma}) \cdot \overline{r}_j,\\
 E[\sum_{i \in F_j} c_{ij}] & \leq & ((1-1/e) \cdot \overline{r}_j) \cdot d_j^{(c)} +
(((1-\frac{1}{e^\gamma})-(1-1/e)) \cdot \overline{r}_j) \cdot  d_j^{(d)}.
 \end{eqnarray*}
\end{lemma}

A rather technical but not difficult proof of the above lemma is given 
in Appendix~\ref{app:exp-conn}.

\subsection{Calculation}
We may now combine the pieces into the algorithm ALG:
\begin{enumerate}
\item solve the LP-relaxation of (\ref{eq1})-(\ref{integrality});
 \item scale the fractional solution as described in Section~\ref{sec_scaling}; \label{step_scaling}
 \item create a family of clusters as described in Section~\ref{sec_clustering}; \label{step_clustering}
 \item round the fractional openings as described in Section~\ref{sec_rounding}; \label{step_rounding}
 \item connect each client $j$ to $r_j$ closest open facilities; \label{step_connection}
 \item output the solution as $(\tilde{x},\tilde{y})$.
\end{enumerate}

\begin{theorem}
 ALG is an $1.7245$-approximation algorithm for FTFL. 
\end{theorem}

\begin{proof}
First observe that the solution produced by ALG is trivially
feasible to the original problem (\ref{eq1})-(\ref{integrality}),
as we simply choose different $r_j$ facilities for client $j$ in step~\ref{step_connection}.
What is less trivial is that all the $r_j$ facilities used by $j$ are within
a certain small distance. Let us now bound the expected connection cost of
the obtained solution.

For each client $j \in \C$ we get $r_j - \overline{r}_j$
facilities opened in Step~\ref{step_scaling}. As we already argued
in Section~\ref{sec_scaling}, we may afford to connect $j$ to these
facilities and pay the connection cost from the difference between
$\sum_i c_{ij}\hat{x}_{ij}$ and $\sum_i c_{ij}\overline{x}_{ij}$.
We will now argue, that client $j$ may connect to the
remaining $\overline{r}_j$ with the expected connection cost bounded by
$\sum_i c_{ij}\overline{x}_{ij}$.

For a special client $j \in (\C \setminus \C')$ we have $\overline{r}_j = 1$ and already in Step~\ref{step_scaling} one special facility at distance $d^{(max)}_j$ from $j$ is opened. We cannot blindly connect $j$ to this facility,
since $d^{(max)}_j$ may potentially be bigger then $\gamma \cdot d_j$. What we do instead
is that we first look at close facilities of $j$ that, as a result of the rounding in Step~\ref{step_rounding},
with a certain probability, give one open facility at a small distance. By Corollary~\ref{cor_1}
this probability is at least $1-1/e$. 
It is easy to observe that the expected connection cost
to this open facility is at most $d^{(c)}_j$. Only if no close facility is open, we use the special facility,
which results in the expected connection cost of client $j$ being at most
\[
(1-1/e)d^{(c)}_j + (1/e)d^{(d)}_j \leq (1-1/e)d^{(c)}_j + (1/e)d_j(1+\frac{R_j}{\gamma-1}) \leq d_j(1+1/(e\cdot(\gamma-1)) \leq  \gamma \cdot d_j,
\]
where the first inequality is a consequence of Lemma~\ref{average_lemma}, and the last one is a consequence of the choice of $\gamma \approx 1.7245$.

In the remaining, we only look at non-special clients $j \in \C'$.
By Lemma~\ref{lem:facility_selection}, client $j$ may select to connect itself to the subset of open facilities
$F_j$, and pay for this connection at most $ ((1-1/e) \cdot \overline{r}_j) \cdot d_j^{(c)} +
(((1-\frac{1}{e^\gamma})-(1-1/e)) \cdot \overline{r}_j) \cdot  d_j^{(d)}$ in expectation.
The expected number of facilities needed on top of those from $F_j$ is $\overline{r}_j - E[|F_j|] = (\frac{1}{e^\gamma} \cdot \overline{r}_j)$. These remaining facilities client $j$ gets deterministically within the distance of at most $3 \cdot d^{(max)}_j$, which is possible by the properties of the rounding procedure described in Section~\ref{sec_rounding},
see Corollary~\ref{cor_rounding}. Therefore, the expected connection cost to facilities not in $F_j$ is at most 
$(\frac{1}{e^\gamma} \cdot \overline{r}_j) \cdot (3 \cdot d^{(max)}_j)$.

Concluding, the total expected connection cost of $j$ may be bounded by
\begin{eqnarray*}
&&((1-1/e) \cdot \overline{r}_j) \cdot d_j^{(c)} +
(((1-\frac{1}{e^\gamma})-(1-1/e)) \cdot \overline{r}_j) \cdot  d_j^{(d)} +
(\frac{1}{e^\gamma} \cdot \overline{r}_j) \cdot (3 \cdot d^{(max)}_j) \\
&\leq&  \overline{r}_j \cdot\left((1-1/e) \cdot d_j^{(c)} +
				((1-\frac{1}{e^\gamma})-(1-1/e)) \cdot d_j^{(d)} +
				\frac{1}{e^\gamma} \cdot (3 d^{(d)}_j) \right)\\
&=&     \overline{r}_j \cdot\left((1-1/e) \cdot d_j^{(c)}
			 	+ ((1+\frac{2}{e^\gamma}) - (1-1/e) ) \cdot d_j^{(d)} \right)\\
&\leq&  \overline{r}_j \cdot\left((1-1/e) \cdot (1-R_j) \cdot d_j
				+ ((1+\frac{2}{e^\gamma}) - (1-1/e) ) \cdot  (1+ \frac{R_j}{\gamma -1}) \cdot d_j\right)\\
&=& \overline{r}_j \cdot d_j \cdot\left( (1-1/e) \cdot (1-R_j)
				+ (\frac{2}{e^\gamma} + 1/e) \cdot  (1+ \frac{R_j}{\gamma -1})\right)\\
&=& \overline{r}_j \cdot d_j \cdot\left( (1-1/e) + (\frac{2}{e^\gamma} + 1/e)
				+ R_j \cdot ((\frac{2}{e^\gamma} + 1/e) \cdot \frac{1}{\gamma -1} - (1-1/e)) \right)\\
&=& \overline{r}_j \cdot d_j \cdot\left( 1+ \frac{2}{e^\gamma} + R_j \cdot \left(\frac{(\frac{2}{e^\gamma} + 1/e)}{\gamma -1} - (1-1/e)\right) \right),
\end{eqnarray*}
where the second inequality follows from Lemma~\ref{average_lemma} and the definition of $R_j$.

Observe that for $1 < \gamma < 2$, we have $ \frac{(\frac{2}{e^\gamma} + 1/e)}{\gamma -1} - (1-1/e) > 0$.
Recall that by definition, $R_j \leq 1$; so, $R_j = 1$ is the
worst case for our estimate, and therefore
\[
 \overline{r}_j \cdot d_j \cdot\left( 1+ \frac{2}{e^\gamma} + R_j \cdot \left(\frac{(\frac{2}{e^\gamma} + 1/e)}{\gamma -1} - (1-1/e)\right) \right) \leq
\overline{r}_j \cdot d_j \cdot (1/e + \frac{2}{e^\gamma})(1+\frac{1}{\gamma-1}).
\]
Recall that $\overline{x}$ incurs, for each client $j$, a fractional connection cost $\sum_{i \in \F} c_{ij}\overline{x}_{ij} \geq \gamma \cdot \overline{r}_j \cdot d_j$.
We fix $\gamma = \gamma_0$, such that $\gamma_0 = (1/e + \frac{2}{e^{\gamma_0}})(1+\frac{1}{\gamma_0-1}) \leq 1.7245$.

To conclude, the expected connection cost of $j$ to facilities opened during the rounding procedure
is at most the fractional connection cost of $\overline{x}$. The total connection cost is, therefore,
at most the connection cost of $\hat{x}$, which is at most $\gamma$ times the connection cost of $x^*$.

By property (P1) of dependent rounding, every single facility $i$ is
opened with the probability $\hat{y}_i$, which is at most
$\gamma$ times $y^*_i$.
Therefore, the total expected cost of the solution produced by ALG is at most $\gamma \approx 1.7245$ times
the cost of the fractional optimal solution $(x^*,y^*)$.
\end{proof}

\noindent \textbf{Concluding remarks.}
We have presented improved approximation algorithms for
the metric Fault-Tolerant Uncapacitated Facility Location problem.
The main technical innovation is the usage and analysis of
dependent rounding in this context. We believe that variants of
dependent rounding will also be fruitful in other location problems.
Finally, we conjecture that the approximation threshold for
both UFL and FTFL is the value $1.46\cdots$ suggested by
\cite{DBLP:journals/jal/GuhaK99}; it would be very interesting to
prove or refute this.

\bibliographystyle{abbrv}
\bibliography{ftufl}

\appendix
\begin{center}
\Large{\textbf{Appendix}}
\end{center}

\section{The rounding approach of \cite{DBLP:conf/focs/Srinivasan01}}
\label{app:dep-round}
The dependent-rounding approach of \cite{DBLP:conf/focs/Srinivasan01}
to round a
given $y = (y_1, y_2, \ldots, y_N) \in [0,1]^N$, is as follows. Suppose
the current version of the rounded vector is
$v = (v_1, v_2, \ldots, v_N) \in [0,1]^N$; $v$ is initially $y$.
When we describe the random choice made in a step below, this choice
is made independent of all such choices made thus far.
If all the $v_i$ lie in $\{0,1\}$, we are done, so let us assume that
there is at least one $v_i \in (0,1)$. The first (simple) case
is that there is exactly
one $v_i$ that lies in $(0,1)$; we round $v_i$ in the natural way --
to $1$ with probability $v_i$, and to $0$ with complementary probability
of $1 - v_i$; letting $V_i$ denote the rounded version of $v_i$, we note
that
\begin{equation}
\label{eqn:P1-typeI}
\E[V_i] = v_i.
\end{equation}
This simple step is called a \emph{Type I iteration}, and it
completes the rounding process. The remaining case is that of a
\emph{Type II iteration}: there are at least two components of $v$ that
lie in $(0,1)$. In this case we choose two such components
$v_i$ and $v_j$ with $i \not= j$, arbitrarily.
Let $\epsilon$ and $\delta$ be the positive constants such that:
(i) $v_i + \epsilon$ and $v_j - \epsilon$ lie in $[0,1]$, with at least
one of these two quantities lying in $\{0,1\}$, and
(ii) $v_i - \delta$ and $v_j + \delta$ lie in $[0,1]$, with at least
one of these two quantities lying in $\{0,1\}$. It is easily seen that
such strictly-positive $\epsilon$ and $\delta$ exist and can be easily
computed. We then update $(v_i, v_j)$ to a random pair $(V_i, V_j)$
as follows:
\begin{itemize}
\item with probability $\delta/(\epsilon + \delta)$, set
$(V_i, V_j) := (v_i + \epsilon, ~v_j - \epsilon)$;
\item with the complementary probability of $\epsilon/(\epsilon + \delta)$, set
$(V_i, V_j) := (v_i - \delta, ~v_j + \delta)$.
\end{itemize}

The main properties of the above that we will need are:
\begin{eqnarray}
\Pr[V_i + V_j = v_i + v_j] & = & 1; \label{eqn:sum-preserve} \\
\E[V_i] = v_i & ~\mbox{and}~ & \E[V_j] = v_j; \label{eqn:exp-preserve} \\
\E[V_i V_j] & \leq & v_i v_j. \label{eqn:neg-correl}
\end{eqnarray}

We iterate the above iteration until all we get a rounded vector. Since each
iteration rounds at least one additional variable, we need at most $N$
iterations.

Note that the above description does not specify the order in which
the elements are rounded. Observe that we may use a predefined laminar family $\SSS$ 
of subsets to guide the rounding procedure. That is, we may first apply
Type II iterations to elements of the smallest subsets, then continue
applying Type II iterations for smallest subsets among those still containing more than
one fractional entry, and eventually round the at most one remaining fractional entry
with a Type I iteration. One may easily verify that executing the dependent rounding
procedure in this manner we almost preserve the sum of entries within each of the subsets
from our laminar family.

\section{Proofs of the statements in Section~\ref{sec:dep-round}}
\label{app:proofs}

\begin{proof}
\textbf{(For Theorem~\ref{dep-vs-indep:general})}
Recall that in the dependent-rounding approach, we begin with the vector
$v^{(0)} = (y_1, y_2, \ldots, y_N)$; in each iteration $t \geq 1$, we
start with a vector $v^{(t-1)}$ and probabilistically modify at most two
of its entries, to produce the vector $v^{(t)}$. We define a potential
function $\Phi(v^{(t)})$, which is a random variable that is fully
determined by $v^{(t)}$, i.e., determined by the random choices made
in iterations $1, 2, \ldots, t$:
\begin{equation}
\label{eqn:Phi-def}
\Phi(v^{(t)}) =
\sum_{\ell = 0}^s \lambda_{\ell}
\sum_{A \subseteq S:~ |A| = \ell} \left((\prod_{a \in A} v^{(t)}_a) \cdot
(\prod_{b \in (S - A)} (1 - v^{(t)}_b)) \right).
\end{equation}
Recall that dependent rounding terminates in some $m \leq N$ iterations.
A moment's reflection shows that:
\begin{equation}
\label{eqn:Phi-begin-end}
\Phi(v^{(0)}) = \E[g_{\lambda, S}(\mathcal{R}(y))]; ~
\E[\Phi(v^{(m)})] = \E[g_{\lambda, S}(\hat{y})].
\end{equation}
Our main inequality will be the following:
\begin{equation}
\label{eqn:Phi-grows}
\forall t \in [m], ~\E[\Phi(v^{(t)})] \geq \E[\Phi(v^{(t-1)})].
\end{equation}
This implies that
\[ \E[\Phi(v^{(m)})] \geq \E[\Phi(v^{(0)})] = \Phi(v^{(0)}), \]
which, in conjunction with (\ref{eqn:Phi-begin-end}) will complete
our proof.

Fix any $t \in [m]$, and fix any
choice for the vector $v^{(t-1)}$ that happens with positive probability.
Conditional on this choice, we will next prove that
\begin{equation}
\label{eqn:Phi-grows-one-iter}
\E[\Phi(v^{(t)})] \geq \Phi(v^{(t-1)});
\end{equation}
note that the expectation in the l.h.s.\ is only w.r.t.\ the random choice
made in iteration $t$, since $v^{(t-1)}$ is now fixed. Once we have
(\ref{eqn:Phi-grows-one-iter}), (\ref{eqn:Phi-grows}) follows from Bayes'
Theorem by a routine conditioning on the value of $v^{(t-1)}$.

Let us show (\ref{eqn:Phi-grows-one-iter}). We first dispose of two
simple cases. Suppose iteration $t$ is a Type I iteration, and that
$v^{(t-1)}_i$ is the only component of $v^{(t-1)}$ that lies in
$(0,1)$. Since $\Phi(v^{(t)})$ is a linear function of the random
variable $v^{(t)}_i$, (\ref{eqn:Phi-grows-one-iter}) holds with
equality, by (\ref{eqn:P1-typeI}). A similar argument holds if
iteration $t$ is a Type II iteration in which the
components $v^{(t-1)}_i$ and $v^{(t-1)}_j$ are probabilistically altered in
this iteration, if at most one of $i$ and $j$ lies in $S$.

So suppose iteration $t$ is a Type II iteration, and that both
$i$ and $j$ lie in $S$ (again, $v^{(t-1)}_i$ and $v^{(t-1)}_j$ are
the components altered in this iteration). Let $v_i = v^{(t-1)}_i$ and
$v_j = v^{(t-1)}_j$ for
notational simplicity, and let $V_i$ and $V_j$ denote their respective
altered values. Note that there are deterministic reals
$u_0, u_1, u_2, u_3$ which depend only on the components of
$v^{(t-1)}$ \emph{other than} $v^{(t-1)}_i$ and $v^{(t-1)}_j$, such that
\begin{eqnarray*}
\Phi(v^{(t-1)}) & = & u_0 + u_1 v_i + u_2 v_j + u_3 v_i v_j; \\
\Phi(v^{(t)}) & = & u_0 + u_1 V_i + u_2 V_j + u_3 V_i V_j.
\end{eqnarray*}
Therefore, in order to prove our desired bound
(\ref{eqn:Phi-grows-one-iter}), we have from
(\ref{eqn:exp-preserve}) and (\ref{eqn:neg-correl}) that is it
is
sufficient to show
\begin{equation}
\label{eqn:pair-coeff-neg}
u_3 \leq 0,
\end{equation}
which we proceed to do next.

Let us analyze (\ref{eqn:Phi-def}), the definition of
$\Phi$, to calculate $u_3$.
Let, for $0 \leq \ell \leq s$,
$\alpha_{\ell}$ denote the contribution of the term
\begin{equation}
\label{eqn:Phi-generic-term}
\lambda_{\ell} \cdot
\sum_{A \subseteq S:~ |A| = \ell} \left((\prod_{a \in A} v^{(t)}_a) \cdot
(\prod_{b \in (S - A)} (1 - v^{(t)}_b)) \right)
\end{equation}
to $u_3$; note that
\[ u_3 = \sum_{\ell = 0}^s \alpha_{\ell}. \]
In order to compute the values $\alpha_{\ell}$, it is convenient to
define certain quantities $\beta_r$, which we do next.
Define $T = S - \{i,j\}$, and note
that $|T| = s - 2$. For $0 \leq r \leq s - 2$, define
\[ \beta_r = \sum_{B \subseteq T:~ |B| = r} \left((\prod_{p \in B} v^{(t)}_p)
\cdot
(\prod_{q \in (T - B)} (1 - v^{(t)}_q)) \right). \]
Now, as a warmup, note that $\alpha_0 = \beta_0$ and $\alpha_s = \beta_{s-2}$.
Let us next compute $\alpha_{\ell}$ for $1 \leq \ell \leq s-1$.
The sum (\ref{eqn:Phi-generic-term}) can contribute a
``$v^{(t)}_i \cdot v^{(t)}_j$'' term in three ways:
\begin{itemize}
\item by taking both $i$ and $j$ in the set $A$ in
(\ref{eqn:Phi-generic-term}) -- this is possible only if $\ell \geq 2$ --
with a coefficient of $\lambda_{\ell} \beta_{\ell - 2}$ for the
``$v^{(t)}_i \cdot v^{(t)}_j$'' term;
\item by taking both $i$ and $j$ in the set $S - A$ in
(\ref{eqn:Phi-generic-term}) -- this is possible only if $\ell \leq s-2$ --
with a coefficient of $\lambda_{\ell} \beta_{\ell}$ for the
``$v^{(t)}_i \cdot v^{(t)}_j$'' term; and
\item by taking \emph{exactly one} of $i$ and $j$ in the
set $A$ -- this is possible for any $\ell \in [s-1]$ --
with a coefficient of $-2 \lambda_{\ell} \beta_{\ell-1}$ for the
``$v^{(t)}_i \cdot v^{(t)}_j$'' term (with the factor of $2$ arising from
the choice of $i$ or $j$ to put in $A$).
\end{itemize}

Rearranging the above three items,
the contribution of $\beta_r$ to $u_3$, for $0 \leq r \leq s-2$,
is $\lambda_{r} - 2 \lambda_{r + 1} + \lambda_{r + 2}$. That is,
\[ u_3 = \sum_{r=0}^{s-2}
(\lambda_{r} - 2 \lambda_{r + 1} + \lambda_{r + 2}) \cdot \beta_r. \]
Thus, the hypothesis of the theorem and the fact that all the values
$\beta_r$ are non-negative, together show that $u_3 \leq 0$ as required
by (\ref{eqn:pair-coeff-neg}).
\end{proof}

\begin{proof}
\textbf{(For Theorem~\ref{dep-vs-indep_thm})}
Let $s = |S|$.
The theorem directly follows from property \textbf{(P1)} if either $s \leq 1$
or $k \geq s$, so we may assume that $s \geq 2$ and that $k \leq s - 1$.
Of course, we may also assume that $k \geq 1$.
Note that for any $x \in \{0,1\}^N$,
\begin{eqnarray*}
\min\{k, \mbox{Sum}_S(x)\} & = &
(\sum_{\ell \leq k} \ell \cdot \mathcal{I}(\mbox{Sum}_S(x) = \ell)) +
(\sum_{\ell > k} k \cdot \mathcal{I}(\mbox{Sum}_S(x) = \ell)) \\
& = & g_{\lambda,S}(x),
\end{eqnarray*}
where
\[ \lambda = (0, 1, 2, \ldots, k, k, k , \ldots, k). \]
It is easy to verify that for all
$0 \leq r \leq s-2$,
$\lambda_{r} - 2 \lambda_{r + 1} + \lambda_{r + 2} \leq 0$.
(Recall that $1 \leq k \leq s-1$. The sum in the l.h.s.\ is zero for all
$r \not = k-1$, and equals $-1$ for $r = k-1$.
Thus we have the theorem, from
Theorem~\ref{dep-vs-indep:general}.
\end{proof}

\begin{proof}
\textbf{(For Theorem~\ref{dr_thm})}
Let $z_i=\frac{y_i}{k}$.
We prove by induction on $|S|$ that  
\begin{equation}
\E[\min\{k,\rsum_S(\Rc(y))\}] \geq k\Bigl(1-\prod_{i\in S}(1-z_i)\Bigr). \label{pfineq1} 
\end{equation} 
This proves the theorem since the RHS above is at least
$k\bigl(1-\exp(-\sum_{i\in S}z_i)\bigr)=k\bigl(1-\exp(-\rsum_S(y)/k)\bigr)$ 
(since $t\geq 1-\exp(-t)$ for all real $t$).  

We now establish \eqref{pfineq1} by induction on $|S|$. The base case when $|S|=1$ is
trivial. For notational simplicity, suppose that $1\in S$. For $|S|\geq 2$, we have
\begin{eqnarray*}
\E[\min\{k,\rsum_S(\Rc(y))\}] & = &
y_1\Bigl(1+\E[\min\{k-1,\rsum_{S\sm\{1\}}(\Rc(y))\}]\Bigr)+
(1-y_1)\E[\min\{k,\rsum_{S\sm\{1\}}(\Rc(y))\}] \\
& \geq & y_1+\E[\min\{k,\rsum_{S\sm\{1\}}(\Rc(y))\}]\Bigl(y_1\cdot\frac{k-1}{k}+1-y_1\Bigr) \\ 
& = & y_1+\Bigl(1-\frac{y_1}{k}\Bigr)\E[\min\{k,\rsum_{S\sm\{1\}}(\Rc(y))\}] \\
& \geq & k\Bigl(z_1+(1-z_1)\bigl(1-\prod_{i\in S\sm\{1\}}(1-z_i)\bigr)\Bigr)
\ =\ k\Bigl(1-\prod_{i\in S}(1-z_i)\Bigr).
\end{eqnarray*}
\end{proof}

\section{Proof of a bound on the expected connection cost of a client}
\label{app:exp-conn}

\begin{proof}
\textbf{(For Lemma~\ref{lem:facility_selection})}
Given client $j$, fractional facility opening vector $\overline{y}$,
distances $c_{ij}$, requirement $\overline{r}_j$, and facility subsets $C_j$ and $D_j$, we will
describe how to randomly choose a subset of at most $k= \overline{r}_j$ open facilities from $C_j \cup D_j$
with the desired properties.

Within this proof we will assume that all the involved numbers are rational.
Recall that the opening of facilities is decided in a dependent rounding routine,
that in a single step couples two fractional entries to leave at most one of them fractional.

Observe that, for the purpose of this argument, we may split a single facility into many identical copies
with smaller fractional opening. One may think that the input facilities and their original openings
were obtained along the process of dependent rounding applied to the multiple 
``small'' copies that we prefer to consider here. 
Therefore, without loss of generality, we may assume that all the facilities have fractional opening equal $\epsilon$,
i.e., $\overline{y}_i = \epsilon$ for all $i \in C_j \cup D_j$. Moreover, we may assume that sets $C_j$ and $D_j$ are disjoint.

By renaming facilities we may obtain that $C_j = \{1,2, \ldots, |C_j| \}$, $D_j = \{|C_j|+1,\ldots,|C_j|+|D_j|\}$,
and $c_{ij} \leq c_{i'j}$ for all $1 \leq i < i' \leq |C_j|+|D_j|$.

Consider random set $S_0 \subseteq C_j \cup D_j$ created as follows. Let $\hat{y}$ be the outcome of
rounding the fractional opening vector $\overline{y}$ with the dependent rounding procedure,
and define $S_0 = \{i: \hat{y}_i=1, (\sum_{j<i} \hat{y})<k\}$. By Corollary~\ref{cor:dr_cor}, we have that
$\E[|S_0|] \geq k \cdot (1 - \exp(-\mbox{Sum}_{C_j \cup D_j}(\overline{y}) / k))$. Define random set $S_\alpha$ for $\alpha \in (0,|C_j|+|D_j|]$
as follows. For $i=1,2,\ldots \lfloor |C_j|+|D_j|-\alpha \rfloor$ we have $i \in S_\alpha$ if and only if $i \in S_0$.
For $i=\lceil |C_j|+|D_j|-\alpha \rceil$, in case $i \in S_0$ we toss a 
(suitably biased) coin and include
$i$ in $S_\alpha$ with probability $\alpha - \lfloor \alpha \rfloor$. For $i > \lceil |C_j|+|D_j|-\alpha \rceil$ we
deterministically have $i \notin S_\alpha$.   

Observe that $\E[|S_\alpha|]$ is a continuous monotone non-increasing function of $\alpha$,
therefore there exists $\alpha_0$ such that $\E[|S_{\alpha_0}|] = k \cdot (1 - \exp(-\mbox{Sum}_{C_j \cup D_j}(\overline{y}) / k))$.
We fix $F_j = S_{\alpha_0}$ and claim that it has the desired properties.
Clearly, by definition, we have $\E[|F_j|] = k \cdot (1 - \exp(-\mbox{Sum}_{C_j \cup D_j}(\overline{y}) / k)) = (1-\frac{1}{e^\gamma}) \cdot \overline{r}_j$.
We next show that the expected total connection cost between $j$ and facilities in $F_j$ is not too large. 

Let $p_i^\alpha = Pr[i \in S_\alpha]$ and $p'_i=p_i^{\alpha_0}=Pr[i \in F_j]$. 
Consider the cumulative probability defined as $cp_i^\alpha= \sum_{j \leq i} p_j^\alpha$.
Observe that application of Corollary~\ref{cor:dr_cor} to subsets of first $i$ elements of $C_j \cup D_j$
yields $cp_i^0 \geq k \cdot (1 - \exp(-\epsilon i / k))$ for $i=1,\ldots, |C_j|+|D_j|$.
Since $(1 - \exp(-\epsilon i / k))$ is a monotone increasing function of $i$
one easily gets that also $cp_i^{\alpha} \geq k \cdot (1 - \exp(-\epsilon i / k))$ for $\alpha \leq \alpha_0$ and $i=1,\ldots, |C_j|+|D_j|$.
In particular, we get $cp_{|C_j|}^{\alpha_0} \geq k \cdot (1 - \exp(-\epsilon |C_j| / k))$.

Since $(1 - \exp(-\epsilon i / k))$ is a concave function of $i$, we also have 
\begin{eqnarray*}
cp_i^{\alpha_0} &\geq& k \cdot (1 - \exp(-\epsilon i / k))\\
								&\geq& (i/|C_j|) \cdot k \cdot (1 - \exp(-\epsilon |C_j| / k))\\
								& = & (i/|C_j|) \cdot (1-\frac{1}{e}) \cdot \overline{r}_j
\end{eqnarray*}
for all $1 \leq i \leq |C_j|$. Analogously, we get 
\begin{eqnarray*}
cp_i^{\alpha_0} &\geq& (k \cdot (1 - \exp(-\epsilon |C_j| / k)))\\
								&& \; +((i-|C_j|)/|D_j|) \cdot k \cdot \left((1 - \exp(\frac{-\epsilon (|C_j|+|D_j|)}{k}))-(1 - \exp(-\epsilon |C_j| / k))\right)\\
								& = &  \overline{r}_j \cdot (1-\frac{1}{e})  + \overline{r}_j \cdot \left( ((i-|C_j|)/|D_j|)((1-\frac{1}{e^{\gamma}})-(1-\frac{1}{e})) \right)
\end{eqnarray*}
for all $|C_j| < i \leq |C_j| +|D_j|$.

Recall that we want to bound $\E[\sum_{i\in F_j} c_{ij}] = \sum_{i \in C_j \cup D_j}p'_i c_{ij}$.
 From the above bounds on the cumulative probability, we get
that, by shifting the probability from earlier facilities to later ones, one 
may obtain a probability vector $p''$ with $p''_i = 1/|C_j| \cdot ((1-\frac{1}{e}) \cdot \overline{r}_j)$ for all $1 \leq i \leq |C_j|$, and $p''_i = 1/|D_j| \cdot ((1-\frac{1}{e^{\gamma}})-(1-\frac{1}{e})) \cdot \overline{r}_j$ for all $|C_j| < i \leq |C_j|+|D_j|$.
Since connection costs $c_{ij}$ are monotone non-decreasing in $i$, when shifting the probability one never decreases
the weighted sum, therefore
\begin{eqnarray*}
 \E[\sum_{i\in F_j} c_{ij}] &=& \sum_{i \in F_j}p'_i c_{ij} \\
 											&\leq& \sum_{i \in F_j}p''_i c_{ij} \\
 											&=& \sum_{1 \leq i \leq |C_j|} 1/|C_j| \cdot ((1-\frac{1}{e}) \cdot \overline{r}_j) c_{ij}\\
 											&& \; + \sum_{|C_j| < i \leq |C_j|+|D_j|} 
 											   1/|D_j| \cdot (((1-\frac{1}{e^{\gamma}})-(1-\frac{1}{e})) \cdot \overline{r}_j) c_{ij}\\
 											&=& ((1-1/e) \cdot \overline{r}_j) \cdot d_j^{(c)} +
(((1-\frac{1}{e^\gamma})-(1-1/e)) \cdot \overline{r}_j) \cdot  d_j^{(d)}.
\end{eqnarray*}

\end{proof}

\end{document}